\documentclass{article}
\usepackage{stmaryrd,amssymb,amsmath,amsthm,url}
\title{Square root meadows\thanks{Partially supported by the 
	  Dutch NWO Jacquard project \emph{Symbiosis}, 
	  project number 638.003.611.
	  In the context of Symbiosis we investigate equational specifications of data types for financial budgets. 
	  This leads to Tuplix Calculus~\cite{BPvdZ}, which makes essential use of meadows. 
	  But financial mathematics uses more operators than those named in the meadow signature. For instance the 
	  definition of volatility makes use of a square root operator, which, 
	  if only for for that reason, enters the operator set needed to specify financial matters.}}

\author{
	Jan A.\ Bergstra
    \and
	Inge Bethke\\
\\
  {\small
	  Section Software Engineering,
	  Informatics Institute,
	  University of Amsterdam}\\
	{\small URL: \url{www.science.uva.nl/~{inge,janb}}
	}
}
\date{}

\newcommand{\SA}{\axname{Signs}}
\newcommand{\SR}{\axname{SquareRoots}}
\newcommand{\Mod}{\axname{Mod}}

\newcommand{\sg}{\operatorname{{\mathbf s}}}
\newtheorem{theorem}{Theorem}

\newtheorem{corollary}{Corollary} 
\newtheorem{definition}{Definition}   
\theoremstyle{definition}

\newcommand{\IL}{\axname{IL}}

\newcommand{\Md}{\axname{Md}}

\newcommand{\Ril}{\axname{RIL}}

\newcommand{\axname}[1]{\ensuremath{\textit{#1}}}

\newcommand{\wortel}[1]{\sqrt[\_]{#1}}

\begin{document}

\maketitle

\begin{abstract}
Let $\mathbb{Q}_0$ denote the rational numbers expanded to a \emph{meadow} by
totalizing inversion such that 
$0^{-1}=0$. $\mathbb{Q}_0$ can be expanded by a total sign function $\sg$ that extracts the sign of a 
rational number.
In this paper we discuss an extension $\mathbb{Q}_0(\sg ,\wortel{\  \ })$ of the signed rationals in which every number has a unique
square root.
\end{abstract}

\section{Introduction}
\label{sec:Intro}
This paper is a contribution to the algebraic specification of number systems. Advantages and disadvantages of the 
algebraic specification of abstract data types have been amply discussed in the computer science literature. We do not 
add anything
new to these matters here but refer the reader to Wirsing~\cite{W90}, the seminal 1977-paper~\cite{GTWW77} of Goguen et al., and 
the overview in Bj{\o}rner and M.C. Henson~\cite{BH07}. 

The primary algebraic properties of the rational, real and complex numbers are captured by the operations and
axioms of \emph{fields} consisting of the equations that define a commutative ring and two axioms, which are not equations,
that define the inverse operator and the distinctness of the two constants. In particular, fields are \emph{partial} 
algebras|because inversion is undefined at 0|and do not possess an equational axiomatization. They do not constitute 
a variety, i.e., they 
are not closed under products, subalgebras and homomorphic images. In the last 15 years algebraic specification languages with pragmatic ambitions have developed in such a way that partial functions are admitted (see e.g.\ CASL~\cite{ABKKMST02}); nevertheless we feel that the original form of algebraic specifications is still valid for theoretical work because it can lead to more stable and more easily comprehensible specifications.

\emph{Meadows}
originate as the design decision to turn inversion 
(or division if one prefers a binary notation for pragmatic reasons)
into a total operator by means of the assumption that $0^{-1}=0$. 
By doing so the investigation of number systems as abstract data types 
can be carried out within the original framework of algebraic 
specifications without taking any precautions for partial functions 
or for empty sorts. 
The equational specification of the variety of meadows has been proposed by Bergstra, Hirshfeld and Tucker~\cite{BHT08,BT07} and 
has subsequently been
elaborated on in detail in~\cite{BT08}.

Following~\cite{BT07} we write $\mathbb{Q}_0$ for 
the rational numbers expanded to a meadow after 
taking its zero-totalized form. 
The main result
of~\cite{BT07} consists of obtaining an equational initial algebra specification 
of $\mathbb{Q}_0$.
In~\cite{BPvdZ} meadows without proper zero divisors are termed 
\emph{cancellation meadows} and in~\cite{BT08} it is shown that the equational theory of cancellation meadows (there called zero-totalized 
fields) has a finite and complete equational axiomatization.  
In~\cite{BP08} this \emph{finite basis} result is extended to a generic form  
enabling its application to extended signatures. In particular, the equational theory of $\mathbb{Q}_0(\sg)$|the rational numbers 
expanded with a total sign function|is shown to be complete finitely  axiomatizable within equational logic. 
In this paper, we will extend cancellation meadows even further to $\mathbb{Q}_0(\sg ,\  \wortel{\  \ })$|the zero-totalized 
signed prime field with
unique square roots.

The paper is structured as follows: in the next section we recall
the axioms for cancellation meadows and the sign function. In Section~\ref{sec:3} we give a complete axiomatization
for $\mathbb{Q}_0(\sg ,  \wortel{\  \ })$. 
We end the paper with some examples illustrating the usage of signed roots
and some conclusions in Section~\ref{sec:4} and \ref{sec:5}, respectively.

\section{Cancellation meadows}
\label{sec:2}
In this section we introduce cancellation meadows and the sign function,
and represent the \emph{Generic Basis Theorem}
that will be used in Section~\ref{sec:3}. We assume that the reader is familiar with using equations and initial algebra semantics to specify
data types. Some accounts of this are Goguen et al.~\cite{GTWW77}, Kamin~\cite{K79}, Meseguer and Goguen~\cite{MG86}, or Wirsing~\cite{W90}.  The theory of computable fields is surveyed in Stoltenberg-Hansen and Tucker~\cite{ST99}.
Moreover, we use standard notations: typically, we let $\Sigma$ be a signature, $Mod_\Sigma(T)$ the class 
of all $\Sigma$-algebras satisfying all the axioms in a theory $T$, and $I(\Sigma, T)$ the initial $\Sigma$-algebra of the theory $T$.

\begin{table}
\centering
\rule[-2mm]{7cm}{.5pt}
\begin{align*}
	(x+y)+z &= x + (y + z)\\
	x+y     &= y+x\\
	x+0     &= x\\
	x+(-x)  &= 0\\
	(x \cdot y) \cdot  z &= x \cdot  (y \cdot  z)\\
	x \cdot  y &= y \cdot  x\\
	1\cdot x &= x \\
	x \cdot  (y + z) &= x \cdot  y + x \cdot  z\\
	(x^{-1})^{-1} &= x \\
	x \cdot (x \cdot x^{-1}) &= x
\end{align*}
\rule[3mm]{7cm}{.5pt}
\vspace{-5mm}
\caption{The set \Md\ of axioms for meadows}
\label{Md}
\end{table}

In~\cite{BT07} \emph{meadows}
were defined as the members of a variety specified by 12 equations. However, in~\cite{BT08}
it was established that the 10 equations in Table~\ref{Md}
imply those used in~\cite{BT07}. Summarizing,
a {meadow} is a commutative ring with unit equipped with
a total unary inverse operation $(\_)^{-1}$
that satisfies the two equations
\begin{align*}
  (x^{-1})^{-1} &= x,   \\
  x\cdot(x \cdot x^{-1}) &= x, \quad(\Ril)
\end{align*}
and in which $0^{-1}=0$. Here \Ril\ abbreviates 
\emph{Restricted Inverse Law}.  We write \Md\ for the set of
axioms in Table~\ref{Md}.

From the axioms in \Md\ the following identities are derivable:
\begin{align*}
        (1)^{-1} &= 1,\\
	(0)^{-1}  &= 0,\\
	(-x)^{-1} &= -(x^{-1}),\\
	(x \cdot  y)^{-1} &= x^{-1} \cdot  y^{-1},
\\
0\cdot x  &= 0,\\
	x\cdot -y &= -(x\cdot y),\\
	-(-x)     &= x.
\end{align*}

The term \emph{cancellation meadow} is introduced in~\cite{BPvdZ} 
for a zero-totalized field that satisfies the so-called 
``cancellation axiom"
\[
x \neq 0 ~\&~ x\cdot y = x\cdot z ~\longrightarrow~ y=z.\]
An equivalent version of the cancellation axiom that we shall
further use in this paper is the
\emph{Inverse Law} (\IL), i.e., the conditional axiom
\begin{align*}
  x\neq 0 ~\longrightarrow~ x\cdot x^{-1}=1.
 \quad(\IL)
\end{align*}
So \IL\ states that there are no proper zero divisors. 
(Another equivalent formulation of the cancellation property is 
$x\cdot y=0~\longrightarrow~x=0\text{ or }y=0$.)

We write $\Sigma_m=(0,1,+,\cdot,-,^{-1})$ for the
signature of (cancellation) meadows and we shall often write $1/t$ 
or 
$\frac 1 t $
for $t^{-1}$, $tu$ for $t\cdot u$, $t/u$ for $t\cdot 1/u$, 
$t-u$ for $t+(-u)$, 
and freely use numerals and exponentiation with constant
integer exponents. We shall further write
\[1_t\text{ for } \frac t t\qquad\text{and}\qquad0_t\text{ for }1-1_t,
\]
so, $0_0=1_1=1$, $0_1=1_0=0$, and for all terms $t$,
\[0_t+1_t=1.\]
Moreover, from $\Ril$ we get 
\begin{align}
1_x^2=1_x
\end{align}
and therefore also
\begin{align}
0_x^2=(1-1_x)^2=1-2\cdot 1_x + 1_x^2=1-1_x=0_x.
\end{align}

We obtain \emph{signed meadows} by 
extending the signature $\Sigma_m=(0,1,+,\cdot,-,^{-1})$ of meadows with the unary sign function $\sg(\_)$.
We write $\Sigma_{ms}$ for this extended signature, so 
$\Sigma_{ms} = (0,1,+,\cdot,-,^{-1},\sg)$.
The sign function $\sg$
presupposes an ordering $<$ of its domain and is defined as follows:

\[\sg(x)=\begin{cases}
-1&\text{if }x<0,\\
0&\text{if }x=0,\\
1&\text{if }x>0.
\end{cases}\]

One can define $\sg$ in an equational manner
by the set \SA\ of axioms 
given in Table~\ref{t:sign}. 
First, notice that by \Md\ and axiom~\eqref{ax1} 
(or axiom~\eqref{ax2}) we find 
\[\sg(0)=0\quad\text{and}\quad\sg(1)=1.\]
Then, observe that in combination with the inverse law \IL, 
axiom \eqref{ax6} is an equational representation of the conditional
equational axiom
\[\sg(x)=\sg(y)~\longrightarrow ~\sg(x+y)=\sg(x).\]
From \Md\ and axioms
\eqref{ax3}--\eqref{ax6} one 
can easily compute $\sg(t)$ for any closed term $t$. An interesting consequence of $\Md \cup \SA$ is the 
idempotency of $\sg$, i.e.\, $\Md \cup \SA \vdash \sg(\sg(x))=\sg(x)$ (see Proposition 2 in~\cite{BP08}).

\begin{table}[hbtp]
\centering
\rule[-2mm]{8.6cm}{.5pt}
\begin{align}
\label{ax1}
\sg(1_x)&=1_x\\
\label{ax2}
\sg(0_x)&=0_x\\
\label{ax3}
\sg(-1)&=-1\\
\label{ax4}
\sg(x^{-1})&=\sg(x)\\
\label{ax5}
\sg(x\cdot y)&=\sg(x)\cdot \sg(y)\\
\label{ax6}
0_{\sg(x)-\sg(y)}\cdot (\sg(x+y)-\sg(x))&=0
\end{align}
\rule[3mm]{8.6cm}{.5pt}
\vspace{-5mm}
\caption{The set \SA\ of axioms for the sign function}
\label{t:sign}
\end{table}
\noindent

The \emph{finite basis result}
for the equational theory of cancellation meadows is formulated in a generic way so that
it can be used for any expansion of a meadow that satisfies
the propagation properties defined below.

\begin{definition}
Let $\Sigma$ be an extension of $\Sigma_m=(0,1,+,\cdot,-,^{-1})$, the
signature of meadows. Let $E\supseteq \Md$ (with \Md\ the set
of axioms for meadows given in Table~\ref{Md}).
\begin{enumerate}
\item
$(\Sigma,E)$ has the \textbf{propagation property for pseudo units} if for
each pair of $\Sigma$-terms $t,r$ and context $C[~]$,
\[E\vdash 1_t\cdot C[r]=1_t\cdot C[1_t\cdot r].\]
\item
$(\Sigma,E)$ has the \textbf{propagation property for pseudo zeros} if for
each pair of $\Sigma$-terms $t,r$ and context $C[~]$,
\[E\vdash 0_t\cdot C[r]=0_t\cdot C[0_t\cdot r].\]
\end{enumerate}
\end{definition}
Preservation of these propagation properties admits the following 
nice result:
\begin{theorem}
[Generic Basis Theorem for Cancellation Meadows]
\label{st}
If $\Sigma\supseteq \Sigma_m,~ E\supseteq \Md$ and $(\Sigma,E)$ 
has the
pseudo unit and the pseudo zero propagation property,
then $E$ is a basis (a complete axiomatisation) of 
$\Mod_\Sigma(E\cup\IL)$.
\end{theorem}
Bergstra and Ponse~\cite{BP08} proved that $\Md$ and $\Md \cup \SA$ satisfy both propagation properties and
are therefore complete axiomatizations of 
$\Mod_\Sigma(\Md\cup\IL)$ and $\Mod_\Sigma(\Md \cup \SA\cup\IL)$, respectively.

\section{Square root meadows} 
\label{sec:3}
A plausible way to totalize the square root operation is to postulate $\sqrt{-1}=i$ and to abandon the 
domain of signed fields in favour of the complex numbers. Here we choose a different approach by stipulating
$\sqrt{x}=-\sqrt{-x}$ for $x<0$. In order to avoid confusion with the principal square root function we deviate from the standard
notation and introduce the unary operation $\wortel{\_}$ called \emph{signed square root}.
We write $\Sigma_{mss}$ for this extended signature, so 
$\Sigma_{mss} = (0,1,+,\cdot,-,^{-1},\sg,\  \wortel{\ })$,
and define the signed square root operation in an equational manner
by the set \SR\ of axioms 
given in Table~\ref{t:wortel}.

\begin{table}[hbtp]
\centering
\rule[-2mm]{8.6cm}{.5pt}
\begin{align}
\label{ax7}
\wortel{x^{-1}}&=(\wortel{x})^{-1}\\
\label{ax8}
\wortel{x\cdot y}&=\wortel{x}\cdot \wortel{y}\\
\label{ax9}
\wortel{x\cdot x\cdot \sg(x)}&=x\\
\label{ax10}
\sg(\wortel{x}-\wortel{y})&=\sg(x-y)
\end{align}
\rule[3mm]{8.6cm}{.5pt}
\vspace{-5mm}
\caption{The set \SR\ of axioms for the square root}
\label{t:wortel}
\end{table}
\noindent

Some additional consequences of the $\Md\cup\SA\cup \SR$ axioms are these:
\begin{align}
\label{6}
\wortel{\sg(x)}&=\sg(x)\text{ because 
$\wortel{\sg(x)}\begin{array}[t]{l}=\wortel{\sg(xxx^{-1})}=\wortel{\sg(x)\sg(x)\sg(x^{-1})}
=\wortel{\sg(x)\sg(x)\sg(x)}\\=\wortel{\sg(x)\sg(x)\sg(\sg(x))}=\sg(x),\end{array}$}\\
\label{7}
\wortel{1_x}&=1_x\text{ because 
$\wortel{1_x}=\wortel{\sg(1_x)}=\sg(1_x)=1_x$,}\\
\label{7a}
\wortel{0_x}&=0_x\text{ similarly,}\\
\label{8}
\wortel{-x}&=-\wortel{x}\text{
because $\wortel{-x}\begin{array}[t]{l}=\wortel{-1\cdot x}=\wortel{-1}\cdot \wortel{x}=\wortel{\sg(-1)}\cdot\wortel{x}\\
=\sg(-1)\cdot\wortel{x}=-1\cdot x=-x,\end{array}$}\\
\label{9}
\wortel{x^2}&=x\cdot \sg(x)\text{
because $\wortel{x^2}\begin{array}[t]{l}=\wortel{x^2\cdot 1_x}=\wortel{x^2}\cdot 1_x=\wortel{x^2}\cdot \sg(1_x)
=\wortel{x^2}\cdot \sg(x)^2\\=\wortel{x^2}\cdot \wortel{\sg(x)}\cdot \sg(x)=\wortel{x^2\sg(x)}\cdot \sg(x)=
x\cdot \sg(x)
.\end{array}$}
\end{align}

Since $(\Sigma_{mss},\Md \cup \SA \cup \SR)$ satisfies both propagation properties, we can apply Theorem~\ref{st}. 
\begin{corollary}
The set of axioms $\Md\cup \SA\cup \SR$
is a complete axiomatisation of 
$\Mod_{\Sigma_{mss}}(\Md\cup \SA \cup \SR \cup\IL)$. 
\end{corollary}
\begin{proof}
We have to prove that the propagation properties for pseudo units 
and pseudo zeros hold in $\Md\cup \SA \cup \SR$. This follows easily
by a case distinction on the forms that $C[r]$
may take. This case distinction has been performed for $\Md\cup \SA$ in~\cite{BP08}.
As an example we consider here the case
$C[\_]\equiv \wortel{\_}$. Then 
\[
1_t\cdot \wortel{r}=1_t^2\cdot \wortel{r}=1_t\cdot \wortel{1_t} \cdot \wortel{r}=
1_t\cdot \wortel{1_t\cdot r}
\]
by (1) and (14).
The propagation property for pseudo zeros is proved in a similar way applying (2) and (15).
\end{proof}

We denote by $\mathbb{Q}_0(\sg , \wortel{\ })$ the zero-totalized signed prime field that contains
$\mathbb{Q}$ and is closed under $\wortel{\ }$. Note that $\mathbb{Q}_0(\sg , \wortel{\ })$ 
is a computable data type (see e.g. Bergstra and Tucker~\cite{BT95}). This statement still requires an efficient 
and readable proof.

To provide an initial algebra specification for $\mathbb{Q}_0(\sg , \wortel{\ })$ may prove a difficult task. In the much simpler case of $\mathbb{Q}_0$ we know that
$\mathbb{Q}_0\cong I(\Sigma_m,\Md + L_4)$ were $L_n$ is the \emph{Lagrange} equation
\[
\frac{1+ x_1^2 +x_2^2 +\cdots + x_n^2}{1+ x_1^2 +x_2^2 +\cdots + x_n^2}=1.\]
Observe that $\mathbb{Q}_0\not\cong I(\Sigma_m,\Md + L_1)$. Indeed, the totalized Galois field 
$(\mathbb{F}_3)_0\models \Md +L_1$: squares in $(\mathbb{F}_3)_0$ are 0 and 1 and thus $1+x^2\neq 0$ in $(\mathbb{F}_3)_0$ 
from which we infer $(\mathbb{F}_3)_0\models \frac{1+ x^2}{1+ x^2}=1$. If $I(\Sigma_m, \Md +L_1)\cong \mathbb{Q}_0$, then $(\mathbb{F}_3)_0$ is a homomorphic image of $\mathbb{Q}_0$. Thus suppose $\phi:\mathbb{Q}_0 \rightarrow \mathbb({F}_3)_0$ is a homomorphism. 
Then|in $(\mathbb{F}_3)_0$| 
\[
0=\frac{1+1+1}{1+1+1}=\frac{\phi(1+1+1)}{\phi(1+1+1)}=\phi(\frac{1+1+1}{1+1+1})=\phi(1)=1,
\]
which is not the case. This then leaves us with the question as to whether or not $\mathbb{Q}_0\cong I(\Sigma_m,\Md + L_2)$.

If for some prime number $p$ the Diophantine equation $x^2+y^2 \equiv (-1) mod\ p$ has no solution, we have that $(\mathbb{F}_p)_0\models L_2$ and a similar argument establishes that $I(\Sigma_m, \Md+L_2)\not\cong\mathbb{Q}_0$. 
However, the existence of such $p$ is not known to us. 

In any case the initial algebra specification of $\mathbb{Q}_0$ can only be considered \emph{stable} once
\begin{enumerate}
\item it has been shown that $\mathbb{Q}_0\not \cong I(\Sigma_m,\Md + L_n)$ for $n=2,3$, and
\item it has been shown that there exists no finite $\omega$-complete|and hence preferable|spec\-ifi\-cation for $\mathbb{Q}_0$ either.
\end{enumerate} 
What follows from these considerations is that the development of a \emph{definitive} initial algebra specification for $\mathbb{Q}_0(\sg , \wortel{\ })$ will be a proces that takes several stages.  Only an initial step has been taken here and more work lies in the future.

\section{Examples}\label{sec:4}
In this section, we briefly discuss 3 examples in which signed roots can play a role.

In the special theory of relativity one frequently encounters equations of the form
\[
\frac{\sqrt{1+\beta}}{\sqrt{1-\beta^2}}=\frac{1}{\sqrt{1-\beta}}
\]
where $\beta=\frac{v}{c}$ and $v$ is the velocity of a moving light source. This particular equation is defined if 
$|v|<c$ and leads to an undefined expression in the case that $c\leq |v|$. 
The theory of signed square roots offers a total representation of the above equation by
\[
\frac{\wortel{1+\beta}}{\wortel{1-\beta^2}}=\frac{\sg^2(1+\beta)}{\wortel{1-\beta}}.
\]
In  such a way arithmetic laws stemming from the special theory of relativity can be modified in order to be universally valid without implicit
or explicit
assumptions. Notice that $\sqrt{-1}=i$ is not essential for the special theory of relativity: e.g.\ the main formula in the Minkowski space|the pseudo-Euclidean space in which special relativity is most conveniently formulated|is the mathematical 
theorem for angles $\alpha$ in spacetime
\[
v^2=-c^2 \Rightarrow \cos(\alpha/v)=\cosh(\alpha/c)\text{ and } v\sin(\alpha/v)=c\sinh(\alpha/c)
\]
which can be justified from calculations on formal power series without the use of complex numbers. 
A similar observation applies to the area of quantum computing. There complex numbers are used and the equation $i^2=1$, but square roots
are only applied to non-negative numbers.

The theory of signed square roots can be extended to  complex numbers by the axioms given in Table~\ref{t:complex}. 
Here we denote by $\overline{x}$ the complex conjugate
of the complex number $x$ and by $Re(x)$ its real part. This, however,  will require a restriction of  the axioms in Table~\ref{t:wortel} to real numbers|e.g.\ 
Axiom (10) becomes $\wortel{Re(x)\cdot Re(y)}=\wortel{Re(x)}\cdot \wortel{Re(y)}$ etc.

\begin{table}[hbtp]
\centering
\rule[-2mm]{8.6cm}{.5pt}
\begin{align}
\label{ax11}
\sg(x)&=\sg(Re(x))\\
\label{ax12}
\wortel{x}&=\wortel{Re(x)}\\
\label{ax13}
Re(x)&=\frac{1}{2}(x + \overline{x})
\end{align}
\rule[3mm]{8.6cm}{.5pt}
\caption{The signed square root for complex numbers}\label{t:complex}
\end{table}
\noindent

In \cite{BP08a} meadows equipped with differentiation operators are introduced. Differential meadows can be equipped with a signed 
square root operator by the axioms given in Tabel~\ref{t:diff}. Axiom (22) can actually be derived from Axiom (21) and the 
equational axiomatization of differential meadows. The existence of
non-trivial differential cancellation meadows with signed square roots is not an obvious matter but requires a modification of the 
existence proof given in~\cite{BP08a}.

\begin{table}[hbtp]
\centering
\rule[-2mm]{8.6cm}{.5pt}
\begin{align}
\label{ax14}
\frac{\partial}{\partial x}\sg(y)&=0\\
\label{ax15}
\frac{\partial}{\partial x}\wortel{y}&=\frac{\sg(y)}{2}(\wortel{y})^{-1}\cdot \frac{\partial}{\partial x}y
\end{align}
\rule[3mm]{8.6cm}{.5pt}
\caption{The signed square root for differential meadows}
\label{t:diff}
\end{table}
\noindent

\section{Conclusion}\label{sec:5}
In this paper we introduced square root meadows. We provided a finite axiomatization for cancellation meadows
expanded with signed square roots and proved its completeness using the Generic Basis Theorem. In addition, 
we gave a few  examples where the theory of signed square roots can make a contribution.
A couple of standard questions|for example, the decidability of the equational theory of  
$\mathbb{Q}_0(\sg , \wortel{\ })$|is left for further research. One step in this direction is the construction of a 
complete
term rewrite system that specifies $\mathbb{Q}_0(\sg , \wortel{\ })$| if this exists at all|or anyway an initial
algebra specification.


\begin{thebibliography}{99}
\bibitem{ABKKMST02} 
E. Astesiano, M. Bidoit, H. Kirchner, B. Krieg-Bruckner, P.D. Mosses, D. Sanella, and A. Tarlecki.
\newblock CASL: the Common Algebraic Specification Language. \emph{Theoretical Computer Science}, 286(2),
153--196, 2002.

\bibitem{BHT08}
J.A. Bergstra, Y.~Hirshfeld, and J.V. Tucker.
\newblock Fields, meadows and abstract data types.
In 
Arnon Avron, Nachum Dershowitz, and Alexander Rabinovich (eds.),
\emph{Pillars of Computer Science (Essays Dedicated to Boris (Boaz)
Trakhtenbrot on the Occasion of His 85th Birthday)},
LNCS 4800, 166--178, Springer-Verlag, 2008.

\bibitem{BP08}
J.A. Bergstra and A. Ponse. 
\newblock A generic basis theorem for cancellation meadows.
Available at
arXiv:0803.3969, 2008.

\bibitem{BP08a}
J.A. Bergstra and A. Ponse. 
\newblock Differential meadows.
Available at
arXiv:0804.3336, 2008.

\bibitem{BPvdZ}
J.A. Bergstra, A. Ponse, and M.B. van der Zwaag. 
\newblock Tuplix Calculus. Electronic report PRG0713, Programming Research Group, University of Amsterdam, December 2007. Available at 
\url{www.science.uva.nl/research/prog/publications.html},
and also at arXiv:0712.3423. To appear in \emph{Scientific Annals of Computer Science}.

\bibitem{BT95}
J.A. Bergstra and J.V. Tucker.
\newblock Equational specifications, complete term rewriting systems, and computable and semicomputable algebras.
\emph{J. ACM}, 42(6), 1194--1230,1995.	 

\bibitem{BT07}
J.A. Bergstra and J.V. Tucker.
\newblock The rational numbers as an abstract data type.
\emph{J. ACM}, 54(2), Article No.\ 7, 2007.	 

\bibitem{BT08}
J.A. Bergstra and J.V. Tucker.
\newblock Division safe calculation in totalised fields.
\newblock \emph{Theory Comput. Syst.}, 43:410--424, 2008.

\bibitem{BH07}
D. Bj{\o}rner and M.C. Henson (editors).
\newblock \emph{Logics of Specification Languages}.
Monographs in Theoretical Computer Science, an EATCS Series.
Springer-Verlag, 2007.

\bibitem{GTWW77}
J.A. Goguen, J.W. Thatcher, E.G. Wagner, and J.B. Wright.
\newblock Initial Algebra Semantics and Continuous Algebras.
\emph{J. ACM}, 24(1):68--95, 1977.

\bibitem{K79}
S. Kamin.
\newblock Some definitions for algebraic data type specifications.
\emph{SIGPLAN Not.}, 14(3):28, 1979.

\bibitem{MG86}
J. Meseguer and J.A. Goguen.
\newblock Initiality, induction, and computability. In: Nivat, M. (ed.),
\emph{Algebraic Methods in Semantics}, Cambridge University Press, Cambridge, 459--541, 1986.

\bibitem{ST99}
V. Stoltenberg-Hansen and J. Tucker., 1999.
\newblock Computable rings and fields. In: Griffor, E. (ed.), \emph{Handbook of Computability Theory}, 
Elsevier, Amsterdam, 363--447,1999.

\bibitem{W90}
M. Wirsing. 
\newblock Algebraic Specification. 
In J. van Leeuwen (ed.), 
\emph{Handbook of Theoretical Computer Science} Volume B (Formal Models and Semantics), 
Elsevier, 675--788, 1990. 	


\end{thebibliography}
\end{document}